\documentclass[conference]{IEEEtran}

\usepackage{amsmath}
\usepackage{amssymb}
\usepackage{amsfonts}
\usepackage{amscd}
\usepackage{amsthm}
\usepackage[noadjust]{cite}
\usepackage{threeparttable}
\usepackage{textcomp}
\usepackage{multirow}

\usepackage{mathrsfs}
\usepackage{bbm}
\usepackage{amsbsy}

\usepackage{paralist}

\usepackage{xfrac}

\newtheorem{theorem}{Theorem}
\newtheorem{lemma}{Lemma}

\newtheorem{corollary}{Corollary}

\theoremstyle{definition}
\newtheorem{example}{Example}
\newtheorem{definition}{Definition}

\makeatletter
\renewcommand*\env@matrix[1][c]{\hskip -\arraycolsep
  \let\@ifnextchar\new@ifnextchar
  \array{*\c@MaxMatrixCols #1}}
\makeatother
\usepackage{graphicx}

\usepackage{etoolbox}
\AtEndEnvironment{example}{\null\hfill\IEEEQED}%

\newcommand{\Fb}{\mathbbmss{F}}

\newcommand{\p}{\pmb}
\newcommand{\xb}{\pmb{x}} 
\newcommand{\cb}{\pmb{c}}

\newcommand{\Ac}{\mathcal{A}} 
 
\newcommand{\Vc}{\mathcal{V}} 
\newcommand{\Ec}{\mathcal{E}}

\newcommand{\Gcpu}{\bar{G}_u} 
\newcommand{\Ecpu}{\bar{\mathcal{E}}_u} 

\newcommand{\Ef}{\mathfrak{E}}
\newcommand{\Df}{\mathfrak{D}}

\newcommand{\betaopt}{{\beta}_{\sf opt}}

\newcommand{\wt}{{\sf wt}}

\title{On Locally Decodable Index Codes}
\author{%
  \IEEEauthorblockN{Lakshmi Natarajan}
  \IEEEauthorblockA{Department of Electrical Engineering\\
                    Indian Institute of Technology Hyderabad\\
                    Sangareddy 502\,285, India\\
                    Email: lakshminatarajan@iith.ac.in}
  \and
  \IEEEauthorblockN{Prasad Krishnan and V.\ Lalitha}
  \IEEEauthorblockA{Signal Processing \& Communications Research Center\\
                    International Institute of Information Technology Hyderabad\\ 
                    Hyderabad 500\,032, India\\
                    Email: \{prasad.krishnan, lalitha.v\}@iiit.ac.in}
}

\begin{document}

\maketitle

\begin{abstract}
Index coding achieves bandwidth savings by jointly encoding the messages demanded by all the clients in a broadcast channel. The encoding is performed in such a way that each client can retrieve its demanded message from its side information and the broadcast codeword. In general, in order to decode its demanded message symbol, a receiver may have to observe the entire transmitted codeword. 
Querying or downloading the codeword symbols might involve costs to a client -- such as network utilization costs and storage requirements for the queried symbols to perform decoding. 
In traditional index coding solutions, this client aware perspective is not considered during code design. As a result, for these codes, the number of codeword symbols queried by a client per decoded message symbol, which we refer to as `locality', could be large. 
In this paper, considering locality as a cost parameter, we view index coding as a trade-off between the achievable broadcast rate (codeword length normalized by the message length) and locality, where the objective is to minimize the broadcast rate for a given value of locality and vice versa. 
We show that the smallest possible locality for any index coding problem is 1, and that the optimal index coding solution with locality 1 is the coding scheme based on fractional coloring of the interference graph. 
We propose index coding schemes with small locality by covering the side information graph using acyclic subgraphs and subgraphs with small minrank. We also show how locality can be accounted for in conventional partition multicast and cycle covering solutions to index coding. 
Finally, applying these new techniques, we characterize the locality-broadcast rate trade-off of the index coding problem whose side information graph is the directed 3-cycle.
\end{abstract}

\section{Introduction} 
The fundamental communication problem in broadcast channels is to design an efficient coding scheme to satisfy the demands of multiple clients with minimal use of the shared communication medium.
Remarkable savings in the broadcast channel use is possible if clients or receivers have prior information stored in their caches that are demanded by other users in the network. 
This could happen, for instance, when the clients are allowed to listen to prior transmissions from the server. 
For such broadcast channels with \textsl{side-information}, Index Coding was proposed by Birk and Kol in \cite{BiK_INFOCOM_98}. 
The idea of Index Coding is to broadcast a coded version of information symbols from the server so that all the receivers can simultaneously decode their demands from the broadcast codeword and their individual side-information symbols. 

Of the several classes of index coding problems discussed in the literature since~\cite{BiK_INFOCOM_98}, the most widely studied is \textit{unicast index coding}, in which each message symbol available at the server is demanded by a unique client.
The side information configuration of a unicast index coding problem is often represented using a directed graph called the \emph{side-information graph} of the problem. 
Given such a side-information graph, the goal of index coding is to design optimal index codes in terms of minimizing the channel usage. 
In \cite{YBJK_IEEE_IT_11}, the authors connected the optimal (scalar) linear index coding problem to finding a quantity called \textit{minrank} of the {side-information graph}. 
This minrank problem is known to be NP-hard in general~\cite{Pee_96}, but several approaches have been taken to address this problem, most popularly via graph theoretic ideas, to bound the optimal index coding rate from above and below; see, for example, \cite{BiK_INFOCOM_98,YBJK_IEEE_IT_11,ALSWH_FOCS_08,BKL_arxiv_10,CASL_ISIT_11,TDN_ISIT_12,BKL_IT_13,NTZ_IT_13,SDL_ISIT_13,SDL_ISIT_14,ArK_ISIT_14,MCJ_IT_14,AgM_arXiv_16}. 
The techniques used in these works naturally lead to constructions of (scalar and vector) linear index codes. 

Most of the known constructions in index coding literature assume that the clients will download all the symbols in the transmitted codeword in order to decode their demanded symbol. 
In practice, this could be prohibitive, since the length of the codeword can be much larger than that of the message demanded at a receiver, especially when the number of users is large and their demands are varied. We may however expect this situation to be the norm in our current and future wireless networks. 
In light of this, it may be appropriate to use
index codes in which each client can decode its demand by accessing only a subset of the codeword symbols. 
Using the terminology from~\cite{HaL_ISIT_12}, we refer to such codes as \emph{locally decodable index codes}.
Designing a locally decodable index code can be thought of as a `client aware' approach to the broadcast problem that takes into account the overhead incurred by the clients while participating in the communication protocol, while conventional index coding (without locality considerations) is more `channel centric' with its emphasis purely on minimizing the number of channel uses. 
Locally decodable index codes have been discussed briefly in~\cite{HaL_ISIT_12} and in the context of privacy in~\cite{KSCF_arXiv_17}. However, a fundamental treatment of the same is not available in the literature to the best of our knowledge. 

\subsection{Contributions and Organization}

In this work, we present a formal structure to the discussion regarding locally decodable index codes and present several constructions of such codes along with their locality parameters for the class of unicast index coding. 
We first define the \emph{locality} of an index code as the number of codeword symbols queried by a receiver for each demanded information symbol, and pose the index coding problem as that of minimizing the broadcast rate (ratio of codeword length to message length) for a given desired value of locality (Section~\ref{sec:defn}). 
We then show that the minimum locality of any unicast index coding problem is $1$, and the optimum index code for this value of locality is the code derived from the optimum fractional coloring of the \emph{interference graph} of the given problem (Section~\ref{sec:fractional_coloring}). 
We then provide several constructions of locally decodable index codes by covering the side-information graph using acyclic subgraphs and subgraphs of small minrank. We also show how the traditional partition multicast~\cite{BiK_INFOCOM_98,TDN_ISIT_12} and cycle covering~\cite{NTZ_IT_13,CASL_ISIT_11} solutions to index coding can be modified to yield locally decodable index codes (Section~\ref{sec:design}).
Using these coding techniques and information-theoretic inequalities we derive the exact trade-off between locality and broadcast rate of the $3$-user unicast index coding problem whose side information graph is a directed cycle (Section~\ref{sec:3cycle}).
Finally in Section~\ref{sec:conclusion} we discuss the relation of locally decodable index codes with the problems of sparse representation of vectors and privacy-preserving index codes and conclude the paper.

{\em Notation:} For any positive integer $N$, we will denote the set $\{1,\dots,N\}$ by $[N]$. Row vectors will be denoted by bold lower case letters, for example $\xb$. For any length $N$ vector \mbox{$\xb=(x_1,\dots,x_N)$} and a set $R \subset [N]$, we define $\xb_R$ to be the sub-vector $(x_j, j \in R)$. For a set of vectors $\xb_1,\dots,\xb_N$ and a subset $K \subset [N]$, we define $\xb_K$ to be the vector $(\xb_j, j \in K)$. The empty set is denoted by $\varnothing$. For a matrix $\p{A}$, the component in $j^{\text{th}}$ row and $i^{\text{th}}$ column is denoted as $\p{A}_{j,i}$.

\section{Definitions and Preliminaries} \label{sec:defn}

The index coding problem~\cite{BiK_INFOCOM_98,YBJK_IEEE_IT_11} consists of a single transmitter jointly encoding $N$ independent messages $\xb_1,\dots,\xb_N$ to broadcast a codeword $\cb$ to multiple receivers through a noiseless broadcast channel. We consider the family of \emph{unicast index coding} problems, where each message $\xb_i$ is desired at exactly one of the receivers, denoted as $(i,K_i)$, where \mbox{$K_i \subset [N]$} is the set of indices of the messages already known to the receiver as side information.
Without loss of generality, we assume that $i \notin K_i$.
The {\em side information graph} of this index coding problem is the directed graph $G=(\Vc,\Ec)$, where the vertex set \mbox{$\Vc=[N]$} and the edge set $\Ec=\{(i,j)\,|\, \text{for all }j \in K_i, i \in [N]\}$. Throughout this paper we will consider only unicast index coding problems and denote a problem by its side information graph $G$.

We assume that the messages $\xb_1,\dots,\xb_N$ are vectors of length $m$ over a finite alphabet $\mathcal{A}$, with $|\mathcal{A}| > 1$, and the codeword $\cb$ is a vector of length $\ell$ over the same alphabet, i.e.,
$\xb_i \in \mathcal{A}^m$, $i \in [N]$, and $\cb \in \mathcal{A}^\ell$.
We will assume that the alphabet $\mathcal{A}$ is arbitrary but fixed for a given index coding problem. 

The transmitted codeword $\cb=(c_1,c_2,\dots,c_\ell)$ is a function of the messages $\cb=\Ef(\xb_1,\dots,\xb_N)$, where $\Ef: \mathcal{A}^{mN} \to \mathcal{A}^{\ell}$ denotes the encoder.
Instead of observing the entire codeword $\cb$, the $i^{\text{th}}$ receiver observes only a sub-vector \mbox{$\cb_{R_i} = (c_j, j \in R_i)$}, where $R_i \subset [\ell]$, and the receiver desires to estimate the message $\xb_i$ using $\cb_{R_i}$ and its side information $\xb_j$, $j \in K_i$. The decoder at the $i^{\text{th}}$ receiver is a function
\begin{equation*}
\Df_i: \mathcal{A}^{|R_i|} \times \mathcal{A}^{m|K_i|} \to \mathcal{A}^{m},
\end{equation*} 
and the estimate of the message $\xb_i$ at the $i^{\text{th}}$ receiver $(i,K_i)$ is \mbox{$\Df_i(\cb_{R_i},(\xb_j,j \in K_i))$}. 
The tuple $(\Ef,\Df_1,\dots,\Df_N)$ of encoding and decoding functions denotes a {\em valid index code} for the index coding problem represented by the side information graph $G=(\Vc,\Ec)$ if 
$\Df_i(\cb_{R_i},(\xb_j,j \in K_i)) = \xb_i$ 
for all $i \in [N]$ and all $(\xb_1,\dots,\xb_N) \in \mathcal{A}^{mN}$, where $\cb=\Ef(\xb_1,\dots,\xb_N)$.

The receiver $(i,K_i)$ decodes $\xb_i$ using the sub-vector $\cb_{R_i}$, and hence, $|R_i|$ is the number of symbols queried or downloaded by this receiver when using the index code $(\Ef,\Df_1,\dots,\Df_N)$.
For fairness, we normalize $|R_i|$ by the number of symbols $m$ in the desired message $\xb_i$ to define the {\em locality $r_i$ of the $i^{\text{th}}$ receiver} as 
$r_i = {|R_i|}/{m}$, for $i \in [N]$.

\begin{definition}
The {\em locality} or the {\em overall locality} of the index code $(\Ef,\Df_1,\dots,\Df_N)$ 
\begin{equation} \label{eq:defn_r}
 r = \max_{i \in [N]} \, r_i = \max_{i \in [N]} \, \frac{|R_i|}{m}
\end{equation} 
is the maximum number of coded symbols queried by any of the $N$ receivers per decoded information symbol.
\end{definition}

The {\em broadcast rate} of this index code is \mbox{$\beta=\ell/m$}, and measures the bandwidth or time required by the source to broadcast the coded symbols to all the receivers. 

Without loss of generality we will consider only those index codes for which $R_1 \cup \cdots \cup R_N = [\ell]$ since any codeword symbol $c_j$ which is not utilized at any of the receivers, i.e., $c_j$ for $j \in [\ell] \setminus (R_1 \cup \cdots \cup R_N)$, need not be generated or transmitted by the encoder.

For a given locality $r$, it is desirable to use an index code with as small a value of $\beta$ as possible and vice versa, which leads us to the following definition.

\begin{definition}
Given a unicast index coding problem $G$, the {\em optimal broadcast rate function} $\beta_{G}^*(r)$ is the infimum of the broadcast rates among all message lengths \mbox{$m \geq 1$} and all valid index codes with locality at the most $r$. 
\end{definition}

The function $\beta_G^*(r)$ captures the trade-off between the reduction in the number of channel uses possible through coding and the number of codeword symbols that a receiver has to query to decode each message symbol. 

We will rely on information-theoretic inequalities to obtain bounds on $\beta_G^*(r)$, and to do so we will assume that the messages $\xb_1,\dots,\xb_N$ are random, independent of each other and are uniformly distributed in $\mathcal{A}^m$. The logarithms used in measuring mutual information and entropy will be calculated to the base $|\mathcal{A}|$. For example, the entropy of $\xb_i$ is $H(\xb_i)=m$ since $\xb_i$ is uniformly distributed in $\mathcal{A}^m$. 

We will now prove some properties of $\beta_G^*(r)$.

\begin{lemma}
The locality $r$ of any valid index coding scheme satisfies \mbox{$r \geq 1$}. 
\end{lemma}
\begin{proof}
This can be shown using the fact $I(\xb_i;\cb_{R_i},\xb_{K_i})=H(\xb_i)=m$, which arises when considering the decoder at the $i^{\text{th}}$ receiver. Since $\xb_i$ is independent of $\xb_{K_i}$ (because $i \notin K_i$), $I(\xb_i;\xb_K)=0$, and
\begin{align}
m &= I(\xb_i;\cb_{R_i},\xb_{K_i}) = I(\xb_i;\xb_{K_i}) + I(\xb_i;\cb_{R_i}|\xb_{K_i}) \nonumber \\
  &= I(\xb_i;\cb_{R_i}|\xb_{K_i}) \leq H(\cb_{R_i}) \leq |R_i|. \label{eq:simple_bound_on_r}
\end{align} 
Hence, $r_i = |R_i|/m \geq 1$ for all $i \in [N]$, and $r \geq 1$. 
\end{proof}

Note that uncoded transmission, i.e., $\cb=(\xb_1,\dots,\xb_N)$, is a valid index code with $r=1$. Hence, we will assume that the domain of the function $\beta_G^*$ is the interval $1 \leq r < \infty$. 

\begin{lemma} \label{lem:convexity}
The function $\beta_G^*(r)$ is convex and non-increasing. 
\end{lemma}
\begin{proof}
See Appendix~\ref{app:lem:convexity}.
\end{proof}

For any valid index code, we have $|R_i| \leq \ell$, and hence, $r = \max_i |R_i|/m \leq \ell/m = \beta$. Hence, if there exists a valid index code with broadcast rate $\beta$, then its locality is at the most $\beta$, and hence, 
\begin{equation} \label{eq:betaG_simple_1}
\beta_G^*(\beta) \leq \beta. 
\end{equation} 
We will denote by $\betaopt(G)$ the infimum among the broadcast rates of all valid index codes for $G$ (considering all possible message lengths $m \geq 1$ and all possible localities $r \geq 1$). Then it follows that $\beta_G^*(r) \geq \betaopt$ for all $r \geq 1$. Together with~\eqref{eq:betaG_simple_1}, choosing $\beta=\betaopt$, we deduce
\begin{equation} \label{eq:betaG_simple_2}
\beta_G^*(\betaopt) = \betaopt.
\end{equation} 

\begin{example} \label{ex:first}
\begin{figure}[!t]
\centering
\includegraphics[width=2.25in]{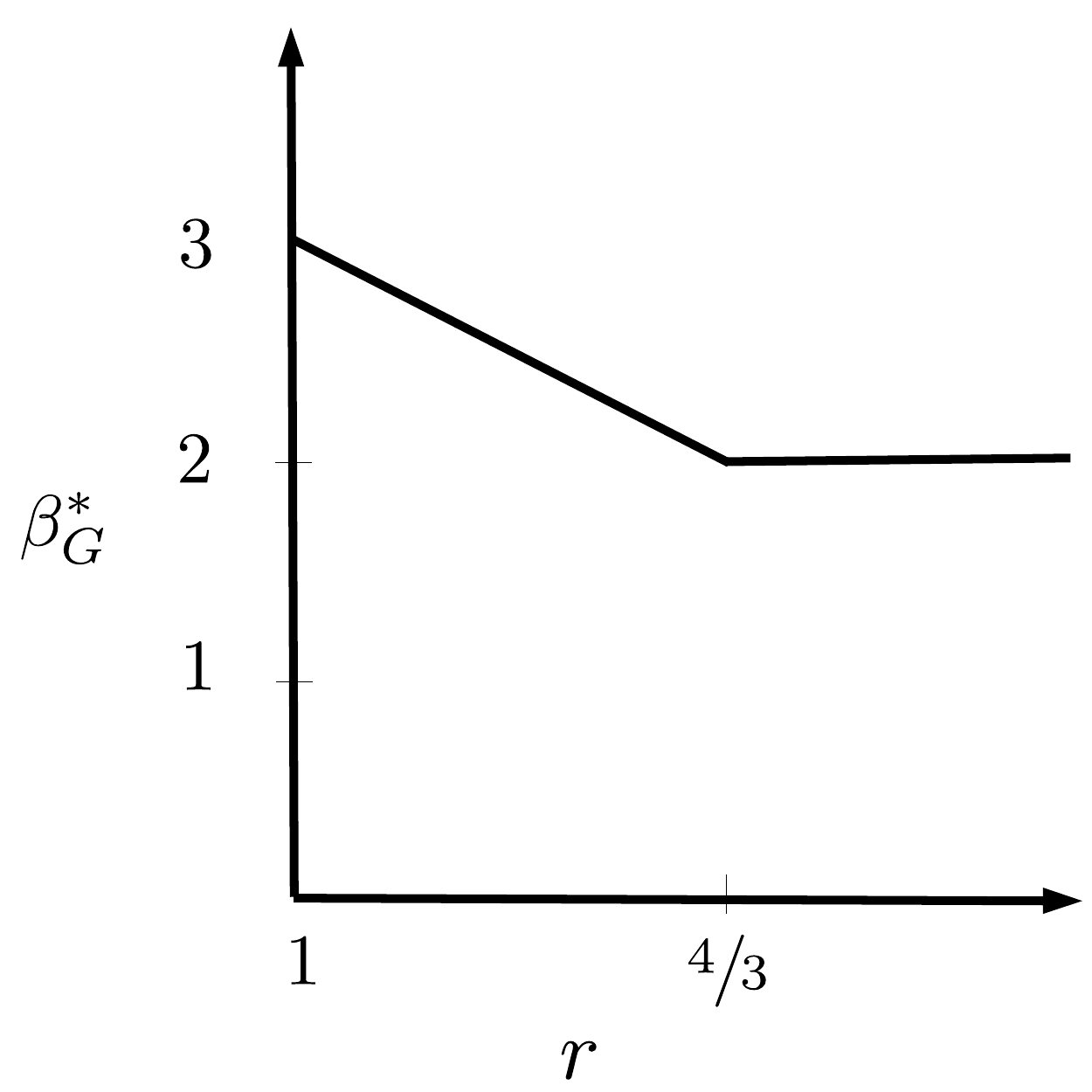}
\caption{The trade-off between the broadcast rate $\beta_G^*(r)$ and the locality $r$ for the index coding problem represented by the $3$-cycle $G$.}
\label{fig:3cycle_tradeoff}
\end{figure} 

\emph{Locality-broadcast rate trade-off of the directed $3$-cycle}.
Let $G$ be the directed $3$-cycle, i.e., $N=3$ and the three receivers $(i,K_i)$ are $(1,\{2\})$, $(2,\{3\})$ and $(3,\{1\})$. 
We show in Section~\ref{sec:3cycle} that for this index coding problem
\begin{equation*} 
 \beta_G^*(r) = \max\{6-3r,2\} \text{ for all } r \geq 1.
\end{equation*} 
This function is shown in Fig.~\ref{fig:3cycle_tradeoff}. We observe that in order to achieve any savings in rate compared to the uncoded transmission ($\beta=N=3$), we necessarily require the locality to be strictly greater than $1$, i.e., each receiver must necessarily query more codeword symbols than the message length to achieve savings in the broadcast channel uses. Also, the smallest locality required to achieve the minimum rate \mbox{$\betaopt=2$} is \mbox{$r=4/3$}.
\end{example}

\section{Fractional Coloring is Optimal for $r=1$} \label{sec:fractional_coloring}


The smallest possible locality for any index coding problem is \mbox{$r=1$}. We will now show that the optimal index coding scheme with locality $r=1$ is the scheme based on fractional coloring of the \emph{interference graph} corresponding to the problem.
We first recall some graph-theoretic terminology related to index coding as used in~\cite{BKL_arxiv_10,SDL_ISIT_13}, and then prove the optimality of fractional coloring in Sections~\ref{sec:r_1_achieve} and~\ref{sec:r_1_converse}. 

\subsection{Preliminaries}

The {\em underlying undirected side information graph} $G_u=(\Vc,\Ec_u)$ of the side information graph $G=(\Vc,\Ec)$ is the graph with vertex set $\Vc=[N]$ and an undirected edge set $\Ec_u=\left\{\,\{i,j\} \, | \, (i,j),\,(j,i) \in \Ec \right\}$, i.e., $\{i,j\} \in \Ec_u$ if and only if $i \in K_j$ and $j \in K_i$. 
%
The {\em interference graph} $\Gcpu=(\Vc,\Ecpu)$ is the undirected complement of the graph $G_u$, i.e., $\Ecpu=\left\{\, \{i,j\} \, | \, \{i,j\} \notin \Ec_u \right\}$. Note that 
\begin{equation} \label{eq:Ecpu}
\{i,j\} \in \Ecpu \text{ if and only if either } i \notin K_j \text{ or } j \notin K_i.
\end{equation} 

For positive integers $a$ and $b$, an \mbox{\em $a:b$ coloring} of the undirected graph $\Gcpu=(\Vc,\Ecpu)$ is a set $\{C_1,C_2,\dots,C_N\}$ of $N$ subsets $C_1,\dots,C_N \subset [a]$, such that $|C_1|=\cdots=|C_N|=b$ and $C_i \cap C_j = \varnothing$ if $\{i,j\} \in \Ecpu$. The elements of $[a]$ are {\em colors}, and each vertex of $\Gcpu$ is assigned $b$ colors such that no two adjacent vertices have any colors in common.
The {\em fractional chromatic number} $\chi_f$ of the undirected graph $\Gcpu$ is
\begin{equation*}
 \chi_f(\Gcpu) = \min \left\{ \frac{a}{b} ~\big\vert~ a:b \text{ coloring of } \Gcpu \text{ exists} \right\}.
\end{equation*} 
The fractional chromatic number is a rational number and can be obtained as a solution to a linear program~\cite{godsil2013algebraic}. The \emph{chromatic number} $\chi(\Gcpu)$ of the graph $\Gcpu$ is the smallest integer $a$ such that an $a:1$ coloring of $\Gcpu$ exists. In general, we have $\chi_f(\Gcpu) \leq \chi(\Gcpu)$.

The main result of this section is
\begin{theorem} \label{thm:frac_coloring}
For any unicast index coding problem $G$, the optimal broadcast rate for locality $r=1$ is $\beta_G^*(1) = \chi_f(\Gcpu)$.
\end{theorem}

\subsection{Proof of achievability for Theorem~\ref{thm:frac_coloring}} \label{sec:r_1_achieve}

It is well known that there exists a coding scheme, called the {\em fractional clique covering}  or {\em fractional coloring} solution, for any index coding problem $G$ with broadcast rate $\beta=\chi_f(\Gcpu)$, see~\cite{BKL_arxiv_10}. 
It is straightforward to observe that $r=1$ for this coding scheme; see Appendix~\ref{app:coloring_achievability} for details. 
It then follows that $\beta_G^*(1) \leq \chi_f(\Gcpu)$.



\subsection{Proof of converse for Theorem~\ref{thm:frac_coloring}} \label{sec:r_1_converse}

Consider any valid index code $(\Ef,\Df_1,\dots,\Df_N)$, possibly non linear, for $G$ with locality $r=1$, message length $m$ and codeword length $\ell$. We will now show that the broadcast rate of this index code is at least $\chi_f(\Gcpu)$.
Since $r=1$, from~\eqref{eq:defn_r} 
and~\eqref{eq:simple_bound_on_r}, we deduce that $|R_1|=\cdots=|R_N|=m$ for this valid index code.

\begin{lemma} \label{lem:basic_lemma_1}
For any $i \in [N]$ and any $P \subset [N]$ such that $i \notin P$, we have 
{(i)} $I(\cb_{R_i};\xb_i|\xb_{K_i \cup P})=m$; 
{(ii)} $H(\cb_{R_i}|\xb_i,\xb_{K_i})=0$; and
{(iii)} $H(\cb_{R_i}|\xb_P)=m$.

\end{lemma}
\begin{proof}
We first observe that
$I(\cb_{R_i};\xb_i|\xb_{K_i \cup P})=H(\xb_i|\xb_{K_i \cup P}) - H(\xb_i|\cb_{R_i},\xb_{K_i \cup P})$. 
Since \mbox{$i \notin K_i \cup P$}, $\xb_i$ is independent of $\xb_{K_i \cup P}$. Also, $\xb_i$ can be decoded using $\cb_{R_i}$ and $\xb_{K_i}$. Hence, $H(\xb_i|\xb_{K_i \cup P})=m$ and $H(\xb_i|\cb_{R_i},\xb_{K_i \cup P})=0$, thereby proving part~\emph{(i)}.

Using the result in part~\emph{(i)} and decomposing the mutual information term $I(\cb_{R_i};\xb_i|\xb_{K_i \cup P})$, we have
\begin{align} \label{eq:key_for_basic_lemma_0}
m = H(\cb_{R_i}|\xb_{K_i \cup P}) - H(\cb_{R_i}|\xb_i,\xb_{K_i \cup P}). 
\end{align} 
Since $\cb_{R_i}$ is a length $m$ vector, we have $H(\cb_{R_i}|\xb_{K_i \cup P}) \leq m$. 
Also, $H(\cb_{R_i}|\xb_i,\xb_{K_i \cup P}) \geq 0$. 
Considering these facts together with~\eqref{eq:key_for_basic_lemma_0}, we deduce that
\begin{equation} \label{eq:key_for_basic_lemma_1}
H(\cb_{R_i}|\xb_i,\xb_{K_i \cup P})=0 \text{ and } H(\cb_{R_i}|\xb_{K_i \cup P}) = m.
\end{equation} 
Observe that~\eqref{eq:key_for_basic_lemma_1} holds for any choice of $P$ such that $i \notin P$. Choosing \mbox{$P=\varnothing$} in the first equality in~\eqref{eq:key_for_basic_lemma_1} proves part~\emph{(ii)} of this lemma.
%
Now using the fact that $\cb_{R_i}$ is of length $m$, and the second equality in~\eqref{eq:key_for_basic_lemma_1}, we have
\begin{align*}
 m \geq H(\cb_{R_i}) \geq H(\cb_{R_i}|\xb_P) \geq H(\cb_{R_i}|\xb_{K_i \cup P}) = m.
\end{align*} 
This shows that $H(\cb_{R_i}|\xb_P)=m$, proving part~\emph{(iii)}.
\end{proof}

\begin{lemma} \label{lem:r_1_R_non_intersecting}
For any $\{i,j\} \in \Ecpu$, we have $R_i \cap R_j = \varnothing$.
\end{lemma}
\begin{proof}
Using~\eqref{eq:Ecpu}, we will assume without loss of generality that $j \notin K_i$. We will now assume that $R_i \cap R_j \neq \varnothing$ and prove the lemma by contradiction. Let \mbox{$t \in R_i \cap R_j$} and \mbox{$P=\{i\} \cup K_i$}. 
From part~\emph{(ii)} of Lemma~\ref{lem:basic_lemma_1}, we have $H(\cb_{R_i}|\xb_i,\xb_{K_i})=0$. In particular, since $t \in R_i$, we have 
\begin{equation} \label{eq:lem:contradiction_1}
H(c_t|\xb_i,\xb_{K_i})=H(c_t|\xb_P)=0.
\end{equation} 

Note that \mbox{$j \notin P$} since \mbox{$j \neq i$} and \mbox{$j \notin K_i$}.
From part~\emph{(iii)} of Lemma~\ref{lem:basic_lemma_1}, we observe that $H(\cb_{R_j}|\xb_P)=m$. This implies that for any given realization of $\xb_P$, the vector $\cb_{R_j}$ is uniformly distributed over $\Ac^m$. Hence, the $m$ coordinates of $\cb_{R_j}$ are independent and uniformly distributed over $\Ac$. Since $t \in R_j$, we conclude that for any given realization of $\xb_P$, $c_t$ is uniformly distributed over $\Ac$, and hence, \mbox{$H(c_t|\xb_P)=1$} which contradicts~\eqref{eq:lem:contradiction_1}. 
\end{proof}

\begin{lemma} \label{lem:r_1_converse}
For any valid index coding scheme for $G$ with $r=1$, the broadcast rate $\beta \geq \chi_f(\Gcpu)$.
\end{lemma}
\begin{proof}
From Lemma~\ref{lem:r_1_R_non_intersecting}, the subsets $R_1,\dots,R_N \subset [\ell]$ are such that $R_i \cap R_j =\varnothing$ if $\{i,j\} \in \Ecpu$ and $|R_i|=m$ for all $i \in [N]$. Hence, $\{R_1,\dots,R_N\}$ is an \mbox{$\ell:m$} coloring of $\Gcpu$. Consequently, the broadcast rate
$\beta = \frac{\ell}{m} \geq \chi_f(\Gcpu)$.
\end{proof}

Combining the converse result in Lemma~\ref{lem:r_1_converse} with the achievability result in Section~\ref{sec:r_1_achieve}, we arrive at Theorem~\ref{thm:frac_coloring}.

\subsection{Corollary and remarks}

Theorem~\ref{thm:frac_coloring} can be easily generalized to the case where the message length $m$ is fixed.

\begin{corollary}
The optimal broadcast rate for index coding problem $G$ with locality $r=1$ and message length $m$ is 
\begin{equation*}
\min \left\{ \frac{a}{m} \, \Big\vert \, \text{an } a:m \text{ coloring of } \Gcpu \text{ exists} \,\right\}.
\end{equation*} 
\end{corollary}
\begin{proof}
The achievability result is similar to the arguments used in Appendix~\ref{app:coloring_achievability} with the additional restriction that the subsets of colors $C_1,\dots,C_N$ are all of size $m$. Converse follows by recognizing that the set of subsets $\{R_1,\dots,R_N\}$ is an $\ell:m$ coloring of $\Gcpu$.
\end{proof}

In~\cite{HaL_ISIT_12} it is remarked that the optimal broadcast rate with $r=1$ among scalar linear index codes over a finite field (i.e., $m=1$, $\mathcal{A}=\Fb_q$ and the encoder $\Ef: \Fb_q^N \to \Fb_q^{\ell}$ is a linear transform) is the chromatic number $\chi(\Gcpu)$. 
Our results in this section provide a strong generalization of this remark.

\section{Designing Locally Decodable Index Codes} \label{sec:design}

We will assume that the alphabet $\Ac$ is a finite field $\Fb_q$ of size $q$. An index code is called {\em vector linear} if the encoder $\Ef:\Fb_q^{mN} \to \Fb_q^{\ell}$ is an $\Fb_q$-linear map. A vector linear index code with $m=1$ is said to be {\em scalar linear}.
First, we will briefly recall the relevant properties of scalar linear index codes from~\cite{YBJK_IEEE_IT_11,DSC_IT_12}, and in the rest of this section we provide constructions of index codes with small locality.

\subsection{Preliminaries}

Given a unicast index coding problem $G$, a matrix $\p{A} \in \Fb_q^{N \times N}$ is said to {\em fit} $G$ if: \emph{(i)}~$\p{A}_{i,i}=1$ for all $i \in [N]$, and \emph{(ii)}~$\p{A}_{j,i}=0$ for all $i \neq j$ such that $j \notin K_i$. 
A matrix $\p{B} \in \Fb_q^{N \times \ell}$ serves as a valid scalar linear encoding matrix for the index coding problem $G$ if and only if there exists an $\p{A} \in \Fb_q^{N \times N}$ such that $\p{A}$ fits $G$ and the column space of $\p{B}$ contains the column space of $\p{A}$ (follows from \cite[Remark~4.6]{DSC_IT_12}). 
The encoder generates the codeword as $\cb=\xb\p{B}$. 
The decoding at the $i^{\text{th}}$ receiver proceeds as follows. Denote the $N$ columns of $\p{A}$ as $\p{a}_1^T,\dots,\p{a}_N^T$, where the superscript $T$ denotes the transpose operation, and each $\p{a}_i$ is a row vector. Since $\p{a}_i^T$ belongs to the column space of $\p{B}$, there exists a vector $\p{d}_i \in \Fb_q^{\ell}$ such that $\p{a}_i^T=\p{B}\p{d}_i^T$. 
The receiver computes $\p{c}\p{d}_i^T$ which equals $\p{xBd}_i^T=\p{xa}_i^T=\sum_{j \in [N]}x_j\p{A}_{j,i}$. Since $\p{A}$ fits $G$, we have $\p{cd}_i^T = x_i + \sum_{j \in K_i}x_jA_{j,i}$. Using the side information, the receiver can recover $x_i$ as $\p{cd}_i^T - \sum_{j \in K_i}x_j\p{A}_{j,i}$.

In order to compute $\p{c}\p{d}_i^T$, the $i^{\text{th}}$ receiver needs to observe only those components of $\p{c}$ which correspond to the non-zero entries of $\p{d}_i^T$. Hence, the locality of the $i^{\text{th}}$ receiver is $r_i=\wt(\p{d}_i)$, i.e., the Hamming weight of $\p{d}_i$. 
If $\p{a}_i^T$ is one of the columns of the encoding matrix $\p{B}$, then $\p{d}_i$ can be chosen such that $\wt(\p{d}_i)=1$ resulting in $r_i=1$. If $\p{B}$ does not contain $\p{a}_i^T$ as one of its columns, then we have the naive upper bound $r_i = \wt(\p{d}_i) \leq \ell$. 

For a given side information graph $G$, the smallest rank among all matrices $\p{A} \in \Fb_q^{N \times N}$ that fit $G$ is called the {\em minrank of $G$ over $\Fb_q$} and is denoted as $\kappa_q(G)$. 
The minimum broadcast rate among scalar linear codes is $\kappa_q(G)$ and can be achieved by using any matrix $\p{B}$ whose columns form a basis of the column space of a rank-$\kappa_q(G)$ matrix $\p{A}$ that fits $G$.


\subsection{Separation based coding scheme}

We first consider a separation based scalar linear index coding technique over $\Fb_q$ for a unicast problem $G$ where the encoder matrix $\p{B}$ is the product of two matrices: an optimal index coding matrix $\p{B}'$ with number of columns equal to $\kappa_q(G)$, and the parity-check matrix $\p{H}$ of a \emph{covering code} $\mathscr{C}$ with co-dimension $\kappa_q(G)$. 
The linear code $\mathscr{C}$ is chosen such that its covering radius is equal to the desired locality $r$, i.e., Hamming spheres of radius $r$ centered around the codewords of $\mathscr{C}$ cover the entire Hamming space. Among all covering codes over $\Fb_q$ with covering radius $r$ and co-dimension $\kappa_q(G)$ we choose the one with the smallest possible blocklength $n_q(r,\kappa_q(G))$.

Following a well known property of covering codes, we observe that any column vector of length $\kappa_q(G)$ over $\Fb_q$ is some linear combination of at the most $r$ columns of $\p{H}$. Thus, if $\p{A}$ is a fitting matrix corresponding to $\p{B}'$, every column of $\p{A}$ can be expressed as a linear combination of at the most $r$ columns of the matrix \mbox{$\p{B}=\p{B}'\p{H}$}. Thus, $\p{B}$ is a valid scalar linear encoder matrix for $G$ with locality $r$ and blocklength $n_q(r,\kappa_q(G))$.

\subsection{Codes from acyclic induced subgraph covers of $G$} \label{sec:ais_cover}

In this section we will provide a technique to construct vector linear codes of small locality by using the acyclic induced subgraphs of $G$. 
For any subset $S \subset [N]$ of vertices of $G$, let $G_S$ denote the subgraph of $G$ induced by $S$.
We will require the following result.

\begin{lemma} \label{lem:ais_lemma}
Let the subgraph $G_S$ of $G$ induced by the subset $S \subset [N]$ be a directed acyclic graph. If there exists is a valid scalar linear encoding matrix for $G$ with codelength $\ell$, then there exists a scalar linear index code with codelength $\ell$ for $G$ such that the locality of every receiver $i \in S$ is $r_i=1$.
\end{lemma} 
\begin{proof}
Let $\p{B} \in \Fb_q^{N \times \ell}$ be a valid encoding matrix. Then there exists a matrix \mbox{$\p{A} =[\,\p{a}_1^T~\cdots~\p{a}_N^T\,] \in \Fb_q^{N \times N}$} that fits $G$ and whose column space is contained in the column space of $\p{B}$. Since $G_S$ is directed acyclic, there exists a \emph{topological ordering} $i_1,\dots,i_{|S|}$ of its vertex set $S=\{i_1,\dots,i_{|S|}\}$, i.e., for any $1 \leq a < b \leq |S|$, there exist no directed edge $(i_b,i_a)$ in $G_S$, and hence, $G$ does not contain the edge $(i_b,i_a)$. It follows that for any choice of $1 \leq a < b \leq |S|$, $i_a \notin K_{i_b}$, and hence, the entry of $\p{A}$ at $i_a^{\text{th}}$ row and $i_b^{\text{th}}$ column is \mbox{$\p{A}_{i_a,i_b}=0$}. Further, for any $1 \leq a \leq |S|$, $\p{A}_{i_a,i_a}=1$ since the diagonal entries of $\p{A}$ are equal to $1$.
Let $\p{E}$ be the $|S| \times |S|$ square submatrix of $\p{A}$ composed of the rows and columns indexed by $S$. It follows that if the rows and columns of $\p{E}$ are permuted according to the topological ordering $i_1,\dots,i_{|S|}$, then $\p{E}$ is lower triangular and all the entries on its main diagonal are equal to $1$. Thus $\p{E}$ is a full-rank matrix, and hence, the columns $\p{a}_{i_1}^T,\dots,\p{a}_{i_{|S|}}^T$ of $\p{A}$ are linearly independent.

Consider a matrix $\p{B}' \in \Fb_q^{N \times \ell}$ constructed as follows. Let the first $|S|$ columns of $\p{B}'$ be $\p{a}_{i_1}^T,\dots,\p{a}_{i_{|S|}}^T$. The remaining $\ell-|S|$ columns of $\p{B}'$ are chosen from among the columns of $\p{B}$ such that the column spaces of $\p{B}'$ and $\p{B}$ are identical. This is possible since $\p{a}_{i_1}^T,\dots,\p{a}_{i_{|S|}}^T$ are linearly independent and are contained in the column space of $\p{B}$.
By construction, the column space of $\p{B}'$ contains the column space of $\p{A}$, and hence, $\p{B}'$ is a valid scalar linear index coding matrix for $G$. Also, for any $i \in S$, the $i^{\text{th}}$ column of $\p{A}$ is a column of $\p{B}'$, and hence, the locality $r_i$ of the $i^{\text{th}}$ receiver is equal to $1$. This completes the proof.
\end{proof}

\begin{definition} 
A set of $M$ subsets $S_1,\dots,S_M \subset [N]$ of the vertex set of the side information graph $G$ is a \emph{$Q$-fold acyclic induced subgraph (AIS) cover} of $G$ if \emph{(i)}~$S_1 \cup \cdots \cup S_M = [N]$, \emph{(ii)}~each $i \in [N]$ is an element of at least $Q$ of the $M$ subsets $S_1,\dots,S_M$, and \emph{(iii)}~all the $M$ induced subgraphs $G_{S_1},\dots,G_{S_M}$ are acyclic.
\end{definition}

Given an AIS cover $S_1,\dots,S_M$ of $G$ and a scalar linear index code of length $\ell$ for $G$, we construct a vector linear code as follows. From Lemma~\ref{lem:ais_lemma}, we know that for each $j \in [M]$, there exists a valid scalar linear encoding matrix $\p{B}_j$ with codelength $\ell$ such that the locality of every receiver $i \in S_j$ is $1$. Consider a vector linear index code that encodes $M$ independent instances of the scalar messages $x_1,\dots,x_N \in \Fb_q$ using the encoding matrices $\p{B}_1,\dots,\p{B}_M$, respectively. The broadcast rate of this scheme is $\ell$. 
If $S_1,\dots,S_M$ is a $Q$-fold AIS cover of $G$, for each \mbox{$i \in [N]$}, there exist $Q$ scalar linear encoders among $\p{B}_1,\dots,\p{B}_M$ that provide locality $1$ at the $i^{\text{th}}$ receiver. The locality provided by the remaining $(M-Q)$ encoders is at most $\ell$ at this receiver. Thus the number of encoded symbols queried by any receiver in the vector linear coding scheme is at the most $Q + (M-Q)\ell$. Normalizing this by the number of message instances $M$, we observe that the locality of this scheme is at the most $\frac{Q + (M-Q)\ell}{M}$. Thus, we have proved 
\begin{theorem}~\label{thm:ais_cover}
If there exists a $Q$-fold AIS cover of $G$ consisting of $M$ subsets of its vertex set, and if there exists a scalar linear index code of length $\ell$ for $G$, then there exists a vector linear code for $G$ with broadcast rate $\ell$, message length $m=M$, and locality at the most $({Q + (M-Q)\ell})/{M}$, and hence,
\begin{equation*}
\beta_G^*\left( \frac{Q + (M-Q)\ell}{M} \right) \leq \ell.
\end{equation*}  
\end{theorem} 

As an application of Theorem~\ref{thm:ais_cover}, consider the following coding scheme. Let the parameter \mbox{$t \geq 1$} be such that the side information graph $G$ contains no cycles of length $t$ or less. With $M=\binom{N}{t}$, let $S_1,\dots,S_M$ be the set of all subsets of $[N]$ of size $t$. The subgraph of $G$ induced by $S_j$, for any $j \in [M]$, is acyclic since $|S_j|=t$. Further, each $i \in [M]$ is an element of 
$Q = \binom{N-1}{t-1}$ 
of the $M$ subsets $S_1,\dots,S_M$. Hence the resulting locality is $({Q + (M-Q)\ell})/{M}$ which can easily be shown to be equal to $(t+(N-t)\ell)/{N}$. Hence, we have 

\begin{corollary}~\label{cor:ais_cover_t}
If $G$ contains no cycles of length $t$ or less, and if there exists a scalar linear index code of length $\ell$ for $G$, then we can achieve broadcast rate $\ell$ with locality $(t+(N-t)\ell)/{N}$.
\end{corollary}

Using Corollary~\ref{cor:ais_cover_t} we can immediately prove the following results.

\begin{lemma} \label{lem:arbitrary_locality}
Let $G$ be any unicast index coding problem with a valid scalar linear index code of length $\ell$. Then there exists a vector linear index code for $G$ with message length $m=N$, broadcast rate $\beta=\ell$ and locality $r=(1+(N-1)\ell)/{N}$.
\end{lemma}
\begin{proof}
Use \mbox{$t=1$} in Corollary~\ref{cor:ais_cover_t}. 
\end{proof}

\begin{lemma} \label{lem:cycles_locality}
If $G$ is a directed cycle of length $N$, then there exists a vector linear coding scheme for $G$ such that the length of each message vector is $m=N$, broadcast rate $\beta=N-1$ and locality $r=2(N-1)/N$.
\end{lemma} 
\begin{proof}
We know that the optimal scalar linear index code for $G$ has length \mbox{$\ell=\kappa_q(G)=N-1$}. Also, with \mbox{$t=N-1$}, $G$ contains no cycles of length $t$ or less. The lemma follows by using Corollary~\ref{cor:ais_cover_t}.  
\end{proof}

When the side information graph is a cycle the optimal broadcast rate $\betaopt=N-1$. Lemma~\ref{lem:cycles_locality} shows that this optimal rate can be achieved with a locality of $2(N-1)/N$, which is strictly less than $2$.

\begin{example} \label{ex:3_cycle_achiev}
Let $G=(\Vc,\Ec)$ be the directed cycle of length $N=3$, i.e, $\Vc=\{1,2,3\}$ and $\Ec=\{(1,2), (2,3), (3,1) \}$. 
From Lemma~\ref{lem:cycles_locality}, $\beta_G^*(4/3) \leq 2$. 

Also, it is well known that the optimal scalar linear index code with $\ell=\kappa_q=2$ is also optimal among all possible index codes for this $G$, i.e., $\betaopt=2$. 
From~\eqref{eq:betaG_simple_2}, $\beta_G^*(2)=2$.
Since $\beta_G^*$ is a non-increasing function, we have
$2 = \beta_G^*(2) \leq \beta_G^*(4/3) \leq 2$,
and hence, $\beta_G^*(4/3)=2$.
\end{example}

\subsection{Codes for symmetric side information problems} \label{sec:minrank_const}

We will now construct vector linear index codes for side information graphs $G$ that satisfy a symmetry property. 
Consider the permutation $\sigma$ on the set $[N]$ that maps $i \in [N]$ to $\sigma(i)=i \mod N + 1$. In this subsection we will assume $G$ to be any directed graph with vertex set $[N]$ such that $\sigma$ is an automorphism of $G$. 
Such unicast index coding problems have been considered before, see~\cite{MCJ_IT_14}, and are related to topological interference management~\cite{Jaf_IT_14}.
First we require the following definition.

\begin{definition}
We say that a set of $M$ matrices $\p{B}_1,\dots,\p{B}_M \in \Fb_q^{N \times \ell}$ is \emph{$(M,Q)$-balanced for $G$} for some integer $Q$ if
\begin{enumerate}[(i)]
\item there exist $M$ matrices $\p{A}_1,\dots,\p{A}_M \in \Fb_q^{N \times N}$ such that for each $i \in [N]$, $\p{A}_i$ fits $G$ and the column spaces of $\p{A}_i$ and $\p{B}_i$ are identical; and
\item for each $i \in [N]$, there exist at least $Q$ distinct indices $j \in [N]$ such that the $i^{\text{th}}$ column of $\p{A}_j$ is a column of $\p{B}_j$.
\end{enumerate} 
\end{definition}

Observe that if $\p{B}_1,\dots,\p{B}_M$ is an $(M,Q)$-balanced set for $G$, then each of these matrices is a valid scalar linear index code for $G$. Also, from the second part of the definition, for at least $Q$ of these codes the locality of the $i^{\text{th}}$ receiver is $1$. 
Now consider the vector linear index coding scheme that is obtained by time-sharing the $M$ codes corresponding to $\p{B}_1,\dots,\p{B}_M$ with equal time shares. Since all the $M$ scalar index codes have broadcast rate $\ell$, the broadcast rate of the overall time-sharing scheme is also $\ell$. For any receiver $i \in [N]$, there exist at least $Q$ scalar codes for which the locality at this receiver is $1$, and for the remaining $(M-Q)$ codes the locality of this receiver is at the most $\ell$. Hence, for the overall time sharing scheme, the locality of any receiver is at the most 
$\frac{Q + (M-Q)\ell}{M} = \ell - \frac{Q(\ell-1)}{M}$,
which is less than the broadcast rate $\ell$. 
We summarize this result as
\begin{theorem} \label{thm:balanced_const}
If there exists a set of $(M,Q)$-balanced $N \times \ell$ matrices for $G$ then 
$\beta_G^*\left( \frac{Q+(M-Q)\ell}{M} \right) \leq \ell$.
\end{theorem}

Based on the above theorem we derive the following result which holds for any index coding problem with symmetric side information graph $G$.

\begin{theorem} \label{thm:cyclic_const}
If the cyclic permutation $\sigma$ is an automorphism of $G$ and if $\kappa_q$ is the minrank of $G$ over $\Fb_q$, then
\begin{equation*}
 \beta_G^*\left( \frac{\kappa_q(N-\kappa_q+1)}{N} \right) \leq \kappa_q.
\end{equation*} 
\end{theorem}
\begin{proof}
Suppose $\p{A} \in \Fb_q^{N \times N}$ fits $G$ is of rank $\ell=\kappa_q(G)$ and let $\p{B}$ be an $N \times \ell$ matrix composed of a set of $\ell$ linearly independent columns of $\p{A}$.
Note that when $\p{B}$ is used as a scalar linear index code there exist $\ell$ receivers with locality $1$ since $\ell$ columns of $\p{A}$ appear as columns of $\p{B}$. 

Let $\p{P}$ be the permutation matrix obtained by cyclically shifting down the rows of the $N \times N$ identity matrix by one position. It is straightforward to verify that $\p{PAP}^T$ fits the graph $\sigma(G)$. 
Since $\sigma$ is an automorphism of $G$, we have $\sigma(G)=G$, and hence, $\p{PAP}^T$ fits $G$. Since the column space of $\p{PB}$ is identical to that of $\p{PAP}^T$, $\p{PB}$ represents a valid scalar linear index code for $G$. 
Using this argument iteratively we deduce that for any $i \in [N]$, the matrix $\p{A}_i=\p{P}^i\p{A}(\p{P}^i)^T$ fits $G$, and the column space of $\p{B}_i=\p{P}^i\p{B}$ is identical to that of $\p{A}_i$.
Further, using a counting argument, we observe that for any $i \in [N]$, there exist $\ell$ distinct values of $j$ such that the $i^{\text{th}}$ column of $\p{A}_j$ is a column of $\p{B}_j$. Hence, $\p{B}_1,\dots,\p{B}_N$ is an $(N,\ell)$-balanced set for $G$, and the statement of this theorem follows from Theorem~\ref{thm:balanced_const}.
\end{proof}

\subsection{Codes from optimal coverings of $G$}

Several index coding schemes in the literature partition the given index coding problem (side information graph $G$) into subproblems (subgraphs of $G$), and apply a pre-defined coding technique on each of these subproblems independently. 
The overall codelength is the sum of the codelengths of the individual subproblems.
The broadcast rate is then reduced by optimizing over all possible partitions of $G$. We will now quickly recall a few such \emph{covering-based} index coding schemes, and then show how they can be modified to guarantee locality.

Let $G$ be the side information graph of any given index coding problem where the side information index set of receiver $i$ is $K_i \subset [N]$. 

\subsubsection*{Partition Multicast}
The \emph{partition multicast} or the \emph{partial clique covering}  scheme uses the transpose of the parity-check matrix of an appropriate maximum distance separable (MDS) code to encode each subgraph of $G$~\cite{BiK_INFOCOM_98,BKL_IT_13,TDN_ISIT_12,ArK_ISIT_14,SDL_ISIT_14}. Specifically let $G_S$ be the subgraph of $G$ induced by the subset of vertices \mbox{$S \subset [N]$}. 
The number of information symbols in the index coding problem $G_S$ is $|S|$, and the side information of receiver \mbox{$i \in S$} in $G_S$ is $K_i \cap S$. 
The partition multicast scheme uses a scalar linear encoder for $G_S$ where the encoding matrix is the transpose of the parity-check matrix of an MDS code of length $|S|$ and dimension $\min_{i \in S}|K_i \cap S|$. This code for $G_S$ encodes messages of length \mbox{$m_S=1$} and has codelength $\ell_S=|S|-\min_{i \in S} |K_i \cap S|$. We will use the trivial value of locality $r_S=\ell_S$ for this coding scheme.
The finite field $\Fb_q$ must be sufficiently large to guarantee that the MDS codes of required blocklength and dimension exist for every possible choice of $S \subset [N]$.

\subsubsection*{Cycle Covering}

The \emph{cycle covering} scheme~\cite{NTZ_IT_13,CASL_ISIT_11} considers subgraphs $G_S$ which form a cycle of length $|S|$.  
If $G_S$ is a directed cycle of length $|S|$, then it is encoded using a scalar linear index code with encoding matrix equal to the transpose of the parity-check matrix of a repetition code of length $|S|$, resulting in message length \mbox{$m_S=1$} and index codelength $\ell_S=|S|-1$. Again, we will use the trivial value of locality $r_S=\ell_S$.
If $G_S$ is not a directed cycle, then the corresponding information symbols are transmitted uncoded resulting in $m_S=1$, $\ell_S=|S|$ and locality $r_S=1$. 

In similar vein to partition multicast and cycle covering schemes, consider the following proposed coding technique that applies the optimal scalar linear index code over each subgraph $G_S$. 

\subsubsection*{Minrank Covering}
Encode each subgraph $G_S$ using its own optimal scalar linear index code. The message length $m_S=1$, codelength $\ell_S$ equals the minrank $\kappa_q(G_S)$ of the subgraph, and the locality $r_S$ equals the codelength $\kappa_q(G_S)$. By partitioning $G$ into subgraphs $G_S$ of small minrank we can achieve a small locality for the overall scheme.

\subsubsection{Scalar-linear codes}
Now consider any covering-based index coding technique (such as partition mutlicast, cycle covering or minrank covering) for $G$. 
Let the scalar linear index code (i.e., message length equal to $1$) associated with the subgraph $G_S$, $S \subset [N]$, have codelength $\ell_S$ and locality $r_S$. 
The overall index code uses a partition of the vertex set $[N]$, which is represented by the tuple $(a_S, S \subset [N])$, where each $a_S \in \{0,1\}$ is such that the partition of $[N]$ consists of all subsets $S$ with $a_S=1$.
Note that $(a_S, S \subset [N])$ represents a partition of $[N]$ if and only if $\sum_{S: i \in S}a_S=1$ for every $i \in [N]$, i.e., every vertex $i$ is contained in exactly one of the subsets of the partition.
The covering-based index coding technique applies an index code of length $\ell_S$ and locality $r_S$ to each subgraph $G_S$ with \mbox{$a_S=1$} independently. Thus the codelength of the overall index code is \mbox{$\ell=\sum_{S: a_S=1}\ell_S=\sum_{S \subset [N]}a_S\ell_S$} and locality is $r=\max_{S:a_S=1}r_S=\max_{S \subset [N]}a_Sr_S$.
By optimizing over all possible partitions of $G$, we have 

\begin{theorem}[Covering with locality] \label{thm:partition_scalar}
Consider a family of scalar linear index codes, one for each $G_S$, $S \subset [N]$, with length $\ell_S$ and locality $r_S$. 
Given any $r \geq 1$, the value of $\beta_G^*(r)$ is upper bounded by the solution to the following integer program
\begin{align*}
&\text{minimize } \sum_{S \subset [N]}a_S\ell_S \text{ subject to} \\
&\sum_{S: i \in S}a_S = 1 ~\forall\, i \in [N], \text{ and }
a_Sr_S \leq r ~\forall\, S \subset [N]\\
&\text{where } a_S \in \{0,1\}.
\end{align*}  
\end{theorem}
The second constraint \mbox{$a_Sr_S \leq r$} in Theorem~\ref{thm:partition_scalar} ensures that the locality of the resulting coding scheme is at the most $r$. Since $a_S \in \{0,1\}$, this implies that when solving for the optimal partition, the integer program considers only those subsets $S$ with $r_S \leq r$, i.e., locality $r$ is achieved by partitioning $G$ into subproblems of small locality.

\subsubsection{Vector-linear codes}

It is known that the linear programming (LP) relaxation, that allows each $a_S$ to assume values in the interval $[0,1]$ instead of \mbox{$a_S \in \{0,1\}$}, provides an improvement in rate for covering-based index coding schemes by time-sharing over several partitions of $G$. This technique too can be adapted to ensure locality. Consider the following vector linear schemes that can be used to cover a given side information graph $G$.

\subsubsection*{Vector Partition Multicast}

For any $S \subset [N]$, the partition multicast scheme uses a scalar code with length $|S|-\min_{i \in S}|K_i \cap S|$ for the subgraph $G_S$. Applying Lemma~\ref{lem:arbitrary_locality}, we obtain a vector linear code for $G_S$ with message length $m_S=|S|$, broadcast rate $\beta_S=|S|-\min_{i \in S}|K_i \cap S|$, codelength $\ell_S=m_S\beta_S$ and locality $r_S=(1+(|S|-1)\beta_S)/{|S|}$. Note that the locality $r_S$ is strictly less than the rate $\beta_S$.

\subsubsection*{Vector Cycle Covering}

In this scheme, if a subgraph $G_S$ is a directed cycle, we use the vector linear index code promised by Lemma~\ref{lem:cycles_locality}, with message length $m_S=|S|$, broadcast rate $\beta_S=|S|-1$, codelength $\ell_S=m_S\beta_S=|S|(|S|-1)$, and locality $r_S=2(|S|-1)/{|S|}$. If $G_S$ is not a cycle, we use uncoded transmission that results in \mbox{$m_S=1$}, \mbox{$\ell_s=|S|$} and \mbox{$r_S=1$}.

\subsubsection*{Vector Minrank Covering} 

For a given $G_S$ we start with the optimal scalar linear code of length $\kappa_q(G_S)$ and use Lemma~\ref{lem:arbitrary_locality} to obtain a vector linear code with message length \mbox{$m_S=|S|$}, rate \mbox{$\beta_S=\kappa_q(G_S)$}, codelength \mbox{$\ell_S=m_S\beta_S$} and locality $r_S={(1+(|S|-1)\kappa_q(G_S))}/{|S|}$. 

Now consider a family of vector linear coding schemes that encodes each subgraph $G_S$, \mbox{$S \subset [N]$}, using a linear code of codelength $\ell_S$ with locality $r_S$, message length $m_S$ and rate $\beta_S=\ell_S/m_S$.
For some choice of integers $k_S$, \mbox{$S \subset [N]$}, perform time-sharing among all the subgraphs of $G$, by encoding $k_S$ independent realizations of the subgraph $G_S$ for each \mbox{$S \subset [N]$}.
In this scheme, the total number of message symbols intended for receiver \mbox{$i \in [N]$} is $\sum_{S: i \in S}k_Sm_S$. The number of codeword symbols $|R_i|$ queried by receiver $i$ is at the most $\sum_{S: i \in S}k_Sr_Sm_S$. The overall length of this index coding scheme is $\ell=\sum_{S \subset [N]}k_S\ell_S$. 
Suppose we require the message length corresponding to every receiver to be identical,
\begin{equation} \label{eq:vector_linear_eq1}
 \sum_{S: i \in S}k_Sm_S = m \text{ for all } i \in [N].
\end{equation} 
Define \mbox{$a_S=k_Sm_S/m$} for all \mbox{$S \subset [N]$}. Then~\eqref{eq:vector_linear_eq1} is equivalent to $\sum_{S: i \in S}a_S=1$ for every \mbox{$i \in [N]$}. The locality of receiver $i$ is $|R_i|/m$ and is upper bounded by $\sum_{S: i \in S}a_Sr_S$. The broadcast rate is $\beta=\sum_{S \subset [N]}k_S\ell_S/m = \sum_{S \subset [N]}a_S\beta_S$.
By optimizing over all possible choices of the time-sharing parameters \mbox{$(a_S, S \subset [N])$}, we arrive at 

\begin{theorem}[Fractional covering with locality] \label{thm:partition_fractional}
Consider any family of vector linear index codes, one for each subgraph $G_S$, $S \subset [N]$, of $G$ with locality $r_S$ and broadcast rate $\beta_S$. For any $r \geq 1$, $\beta_G^*(r)$ is upper bounded by the solution to the following linear program
\begin{align*}
&\text{minimize } \sum_{S \subset [N]}a_S\beta_S \text{ subject to} \\
&\sum_{S: i \in S}a_S = 1  \text{ and }
 \sum_{S: i \in S}a_Sr_S \leq r ~\forall\, i \in [N]\\
&\text{where } a_S \in [0,1]. 
\end{align*}  
\end{theorem}

\section{Locality-Rate Trade-Off of Directed $3$-Cycle} \label{sec:3cycle}


Let $G$ be the directed $3$-cycle i.e., $N=3$ and the three receivers $(i,K_i)$ are $(1,\{2\})$, $(2,\{3\})$ and $(3,\{1\})$. We will now characterize its locality-rate trade-off given by the optimal broadcast rate function $\beta_G^*(r)$ using the achievability schemes of Section~\ref{sec:fractional_coloring}~and~\ref{sec:ais_cover} and a converse based on information inequalities. 
The objective of this section is to prove 
\begin{theorem} \label{thm:3cycle}
For the unicast index coding problem represented by the directed $3$-cycle $G$, the locality-broadcast rate trade-off 
\begin{equation*} 
 \beta_G^*(r) = \max\{6-3r,2\} \text{ for all } r \geq 1.
\end{equation*} 
\end{theorem}
We will now first prove the achievability part which will provide an upper bound on $\beta_G^*(r)$, and then provide an information-theoretic converse to arrive at a lower bound. Theorem~\ref{thm:3cycle} will then follow immediately from these two bounds.

\subsection{Proof of achievability}

From Example~\ref{ex:3_cycle_achiev} in Section~\ref{sec:ais_cover} we know that for $r=4/3$, $\beta_G^*(r)=2$. 
Also, from Theorem~\ref{thm:frac_coloring}, we know that when $r=1$, $\beta_G^*(r)=\chi_f(\Gcpu)$. Since $G$ is a directed $3$-cycle, $\Gcpu$ is a complete graph on three vertices, i.e., an edge exists between every pair of vertices in $\Gcpu$. Hence, $\chi_f(\Gcpu)=3$, and therefore, $\beta_G^*(1)=3$.

Now since $\beta_G^*(r)$ is a convex function, considering the $(r,\beta_G^*)$ points $(1,3)$ and $(4/3,2)$, for any $\alpha \in (0,1)$, we have $\beta_G^*(\alpha + (1-\alpha)4/3) \leq 3\alpha + 2(1-\alpha)$. With $r=\alpha + (1-\alpha)4/3$, we therefore have
\begin{equation} \label{eq:3cycle_achiev_1}
\beta_G^*(r) \leq 6-3r, \text{ for } 1 \leq r \leq 4/3.
\end{equation} 
Further, since $\beta_G^*$ is a decreasing function, for any $r \geq 4/3$, $\beta_G^*(r) \leq \beta_G^*(4/3)=2$. Combining this with~\eqref{eq:3cycle_achiev_1} we have
\begin{equation*}
\beta_G^*(r) \leq \max\{6-3r,2\} \text{ for all } r \geq 1.
\end{equation*} 

\subsection{Proof of converse}

To prove the converse we first require the following general result. 
Let $G$ be any unicast index coding problem involving $N$ messages such that the cyclic permutation $\sigma$ that maps $i \in [N]$ to $i \mod N + 1$ is an automorphism of the side information graph $G$. We now show that we can assume without loss of generality that the index sets of the codeword symbols queried by the $N$ receivers satisfy certain symmetry properties. 

\begin{lemma} \label{lem:cyclic_symmetry}
If the cyclic permutation $\sigma$ is an automorphism of $G$ and if there exists an index code $(\Ef,\Df_1,\dots,\Df_N)$ for $G$ with broadcast rate $\beta$, then there exists an index code $(\Ef',\Df_1',\dots,\Df_N')$ for $G$ with rate $\beta$ such that the index sets of codeword symbols $R_1',\dots,R_N'$ queried by the receivers satisfy the following properties:
\begin{enumerate}[(i)]
\item $|R_1'|=|R_2'|\cdots=|R_N'|$; and 
\item $|R_1' \cap R_2'| = |R_2' \cap R_3'| = \cdots = |R_N' \cap R_1'|$.
\end{enumerate} 
\end{lemma} 
\begin{proof}
See Appendix~\ref{app:lem:cyclic_symmetry}.
\end{proof}

Let us now assume that $G$ is the directed $3$-cycle.
Consider any valid index coding scheme for $G$ with broadcast rate $\beta$ and locality $r$. Let the message length be $m$ and codelength be $\ell$. Using Lemma~\ref{lem:cyclic_symmetry}, we assume without loss in generality that the sets $R_1,R_2,R_3$ satisfy $|R_1|=|R_2|=|R_3|=rm$ and $|R_1 \cap R_2| = |R_2 \cap R_3|=|R_3 \cap R_1|$.

For the sake of brevity, we abuse the notation mildly by using $(i+1)$ to denote the receiver index $(i \mod 3 + 1)$, and similarly use $(i+2)$ in order to denote $(i+1) \mod 3 + 1$. With this notation, for $i=1,2,3$, the side information index set of the $i^{\text{th}}$ receiver is $K_i=\{i+1\}$. Assume, as usual, that the messages $\xb_1,\xb_2,\xb_3$ are independently and uniformly distributed in $\Ac^m$.

Now considering the $i^{\text{th}}$ receiver, we have $I(\xb_i;\cb_{R_i}|\xb_{i+1})=H(\xb_i)=m$. Expanding this term as a difference of conditional entropies, we have
$H(\cb_{R_i}|\xb_{i+1}) - H(\cb_{R_i}|\xb_i,\xb_{i+1}) = m$. Using  this with the upper bound $H(\cb_{R_i}|\xb_{i+1}) \leq H(\cb_{R_i}) \leq |R_i|$, we arrive at
\begin{equation*}
H(\cb_{R_i}|\xb_i,\xb_{i+1}) \leq |R_i| - m.
\end{equation*} 
Using the above inequality and the fact that $\cb_{R_i}$ is a deterministic function of all three messages $\xb_i,\xb_{i+1},\xb_{i+2}$, we have $I(\cb_{R_i};\xb_{i+2}|\xb_i,\xb_{i+1})=H(\cb_{R_i}|\xb_i,\xb_{i+1}) \leq |R_i|-m$. Hence,
\begin{align*}
H(\xb_{i+2}|\xb_i,\xb_{i+1}) - H(\xb_{i+2}|\cb_{R_i},\xb_i,\xb_{i+1}) \leq |R_i| - m.
\end{align*} 
Since $H(\xb_{i+2}|\xb_i,\xb_{i+1})=m$, we obtain the lower bound
\begin{equation} \label{eq:3cycle_conv_1}
H(\xb_{i+2}|\cb_{R_i},\xb_i,\xb_{i+1}) \geq 2m-|R_i|. 
\end{equation}  

Our objective now is to use the above inequality to obtain an upper bound on $|R_i \cap R_{i+2}|$, which can then be translated into a lower bound on $\ell$, and hence, a lower bound on $\beta$. To do so, observe that $\cb_{R_i}$ is composed of the following two sub-vectors $\cb_{R_i \cap R_{i+2}}$ and $\cb_{R_i \setminus R_{i+2}}$. Using~\eqref{eq:3cycle_conv_1}, we obtain
\begin{align*}
H(\xb_{i+2}|\cb_{R_i \cap R_{i+2}},\xb_i) & \geq H(\xb_{i+2}|\cb_{R_i},\xb_i,\xb_{i+1}) \geq 2m-|R_i|.
\end{align*} 
Using this inequality, and the relation $H(\xb_{i+2}|\cb_{R_{i+2}},\xb_i)=0$ (to satisfy the demands of the $(i+2)^{\text{th}}$ receiver), we obtain the following
\begin{align*}
|R_{i+2} \setminus R_i| &\geq H(\cb_{R_{i+2} \setminus R_i}) \\
&\geq I(x_{i+2};\cb_{R_{i+2} \setminus R_i}|\cb_{R_i \cap R_{i+2}},\xb_i) \\
&= H(\xb_{i+2}|\cb_{R_i \cap R_{i+2}},\xb_i) \\ &~~~~~~~~~~~~~~- H(\xb_{i+2}|\cb_{R_{i+2}\setminus R_i},\cb_{R_i \cap R_{i+2}},\xb_i) \\
&= H(\xb_{i+2}|\cb_{R_i \cap R_{i+2}},\xb_i) - H(\xb_{i+2}|\cb_{R_{i+2}},\xb_i) \\
&\geq 2m - |R_i|.
\end{align*} 

Since $|R_1|=|R_2|=|R_3|$, we now have
\begin{align*}
|R_i \cap R_{i+2}| &= |R_{i+2}| - |R_{i+2} \setminus R_i| \\
&\leq |R_i| - (2m - |R_i|) \\
&= 2\left(|R_i|-m \right).
\end{align*} 
Finally, using the fact that $|R_i \cap R_{i+2}|$ is independent of $i$,
\begin{align*}
\ell &= |R_1 \cup R_2 \cup R_3| \\
&\geq \sum_{j=1}^{3} |R_j| - \sum_{i=1}^{3} |R_j \cap R_{j+2}| \\
&= 3|R_i| - 3|R_i \cap R_{i+2}| \\
&\geq 3|R_i| - 3 \times 2(|R_i|-m) \\
&= 6m - 3|R_i|.
\end{align*} 
Dividing both sides by the message length $m$, and remembering that all the receivers have the same locality $r$, we have $\beta=\ell/m \geq 6-3r$. Thus we have 
\begin{align*}
\beta_G^*(r) \geq 6-3r \text{ for all } r \geq 1.
\end{align*} 
Further, $\betaopt(G)=2$, and hence, $\beta_G^*(r) \geq 2$ for all $r \geq 1$. Combining this with the above inequality we have the converse
\begin{equation*} 
\beta_G^*(r) \geq \max\{6-3r,2\} \text{ for all } r \geq 1.
\end{equation*}

\section{Conclusion and Discussion} \label{sec:conclusion}

\subsection*{Relation to other problems}

Locally decodable index codes are related to privacy-preserving broadcasting and the problem of sparse representation of sets of vectors.
In~\cite{KSCF_arXiv_17}, the authors consider scalar linear index codes with locality $r$ under a privacy preserving communication scenario using the terminology \emph{$r$-limited access schemes}. 
The motivation for using $r$-limited access schemes is to reduce the amount of information that any receiver can infer about the demands of other receivers in the network. 
This is achieved by restricting the knowledge of the scalar linear encoder matrix $\p{B}$ at any receiver to at the most $r$ columns instead of revealing the entire matrix. 
Thus each receiver knows the values of at the most $r$ columns of the encoder matrix $\p{B}$, and hence, is required to perform decoding by querying at the most $r$ codeword symbols which correspond to these columns. 

A scalar linear index code with locality $r$ is characterized by a valid encoder matrix $\p{B}$ with a corresponding fitting matrix $\p{A}$ such that any column of $\p{A}$ is some linear combination of at the most $r$ columns of $\p{B}$. Thus the columns of $\p{B}$ serve as an overcomplete basis for a sparse representation of the columns of $\p{A}$.

\subsection*{Conclusion}

We introduced the problem of designing index codes that are locally decodable and have identified several techniques to construct such codes. We have also identified the optimum broadcast rate corresponding to unit locality for any unicast index coding problem, and the complete locality-rate trade-off curve for the problem represented by a directed $3$-cycle. Stronger achievability and/or converse results may be required to gain further insights into the locality-rate trade-off of a general index coding problem.



\section*{Acknowledgment}

The first author thanks Prof.\ B.\ Sundar\ Rajan for discussions regarding the topic of this paper.

\appendices

\section{Proof of Lemma~\ref{lem:convexity}} \label{app:lem:convexity}

The non-increasing property of $\beta_G^*$ follows immediately from its definition. We will use time-sharing to prove convexity. Assume $r_1,r_2 \geq 1$ and let $\epsilon > 0$. For each $j=1,2$, there exists an index code with broadcast rate $\beta_j \leq \beta_G^*(r_j) + \epsilon$ and locality at the most $r_j$. We will denote the blocklength of this code by $\ell_j$, message length by $m_j$ and the subsets of the indices used by the $N$ receivers as $R_{1,j},\dots,R_{N,j}$, where $j=1,2$. For some choice of non-negative integers $k_1$ and $k_2$, consider a time-sharing scheme where the first index code is used $k_1m_2$ times, and the second index code is used $k_2m_1$ times. For this composite scheme, the overall message length is $m=k_1m_2m_1 + k_2m_1m_2 = m_1m_2(k_1+k_2)$. The blocklength is $\ell=k_1m_2\ell_1 + k_2m_1\ell_2$, and the number of codeword symbols utilized by the $i^{\text{th}}$ receiver to decode its desired message is $|R_i|=k_1m_2|R_{i,1}| + k_2m_1|R_{i,2}|$. The locality of this time-sharing system can be upper bounded as
\begin{align*}
\max_i \frac{|R_i|}{m} &= \max_i \,\, \frac{k_1|R_{i,1}|}{(k_1+k_2)m_1} + \frac{k_2|R_{i,2}|}{(k_1+k_2)m_2} \\
&\leq \frac{k_1}{k_1+k_2} r_1 + \frac{k_2}{k_1+k_2}r_2.
\end{align*} 
Similarly, the broadcast rate of this time-sharing scheme can be shown to be equal to $k_1\beta_1/(k_1+k_2) + k_2\beta_2/(k_1+k_2)$, which is upper bounded by
\begin{align*}
\frac{k_1}{k_1+k_2}\beta_G^*(r_1) + \frac{k_2}{k_1+k_2}\beta_G^*(r_2) + \epsilon.
\end{align*} 
Denoting $k_1r_1/(k_1+k_2) + k_2r_2/(k_1+k_2)$ by $r$, and by letting $\epsilon \to 0$, we observe that
\begin{equation*}
\beta_G^*(r) \leq \frac{k_1}{k_1+k_2}\beta_G^*(r_1) + \frac{k_2}{k_1+k_2}\beta_G^*(r_2).
\end{equation*}
Convexity follows by approximating any real number in the interval $(0,1)$ by the rational number $k_1/(k_1+k_2)$ to any desired accuracy by using sufficiently large $k_1$ and $k_2$. 

\section{Fractional Coloring Solution has Unit Locality} \label{app:coloring_achievability}

We will now recall the fractional coloring solution~\cite{BKL_arxiv_10} to an index coding problem $G$ and observe that $r=1$ for this scheme.
Let $C_1,\dots,C_N \subset \{1,\dots,a\}$ be an $a:b$ coloring of $\Gcpu$ such that $\chi_f(\Gcpu)=a/b$. Set codeword length \mbox{$\ell=a$} and message length \mbox{$m=b$}. Denote the components of the message vectors $\xb_i \in \Ac^m$ using the variables $w_{i,t} \in \Ac$ as follows: 
$\xb_i = \left( w_{i,t}, t \in C_i \right)$, i.e., one symbol $w_{i,t}$ corresponding to each color $t$ in the set $C_i$.
Endow the set $\Ac$ with any abelian group structure $(\Ac,+)$.
The symbols of the codeword $\cb=(c_1,\dots,c_\ell)$ are generated as 
$c_t = \sum_{i:\, t \in C_i} w_{i,t}$, for $t \in [\ell]$.

Decoding at the receiver $(i,K_i)$ can be performed as follows. Note that $\xb_i$ is composed of all symbols $w_{i,t}$ such that $t \in C_i$. In order to decode $w_{i,t}$, the receiver retrieves the code symbol $c_t$ which is related to $w_{i,t}$ as 
\begin{equation} \label{eq:frac_color_ach}
 c_t = w_{i,t} + \sum_{\substack{j \neq i \\ j: \, t \in C_j}} w_{j,t}.
\end{equation} 
For any choice of the index $j$ in the summation term above, we have $i \neq j$ and $t \in C_i \cap C_j$. Since $C_1,\dots,C_N$ is a coloring of $\Gcpu$ and $C_i \cap C_j \neq \varnothing$, we deduce that $\{i,j\} \notin \Ecpu$, or equivalently, $j \in K_i$. Hence, for each $j \neq i$ such that $t \in C_j$, the receiver $(i,K_i)$ knows the value of $w_{j,t}$, and thus, can recover $w_{i,t}$ from $c_t$ using~\eqref{eq:frac_color_ach}. Using a similar procedure $(i,K_i)$ can decode all the $m$ symbols in $\xb_i$ from the $m$ coded symbols $(c_t,t \in C_i)$. This decoding method uses $R_i=C_i$ for all $i \in [N]$ and has locality $r=1$.

\section{Proof of Lemma~\ref{lem:cyclic_symmetry}} \label{app:lem:cyclic_symmetry}

Since $\sigma$ is an automorphism of $G$, so is $\sigma^n$ for any $n \in [N]$. Note that the group $\{\sigma,\sigma^2,\dots,\sigma^{N}=1\}$ acts transitively on the vertex set of $G$.
Let $(\Ef,\Df_1,\dots,\Df_N)$ be an index code for $G$ with rate $\beta$ and receiver localities $r_1,\dots,r_N$. We will consider $N$ index coding schemes $(\Ef^{(n)},\Df_1^{(n)},\dots,\Df_N^{(n)})$, $n \in [N]$, each of which is derived from $(\Ef,\Df_1,\dots,\Df_N)$ by permuting the roles of the messages $\xb_1,\dots,\xb_N$. Specifically, the $n^{\text{th}}$ encoder $\Ef^{(n)}$ is the encoder $\Ef$ applied to the $n^{\text{th}}$ left cyclic shift of the message tuple $\xb_1,\dots,\xb_N$, i.e.,
\begin{align*}
\Ef^{(n)}\left( \, (\xb_1,\dots,\xb_N) \,\right) &= \Ef(\, (\xb_{n+1},\xb_{n+2},\dots,\xb_N,\xb_1,\dots,\xb_{n}) \,) \\
&= \Ef\left( \, (\xb_{\sigma^{n}(1)},\dots,\xb_{\sigma^{n}(N)}) \, \right).
\end{align*}
For any $i \in [N]$, in the above expression of $\Ef^{(n)}$, the message $\xb_i$ is the $(i-n)_N^{\text{th}}$ argument of $\Ef$ where $(i-n)_N=(i-n)$ if $(i-n) \geq 1$ and $(i-n)_N = i-n+N$ otherwise. Hence, the encoding function $\Ef^{(n)}$ operates on the message $\xb_i$ in the same manner as$\Ef$ operates on the message $\xb_{(i-n)_N}$. Using the fact that $\sigma$ is an automorphism of $G$, it is easy to see that, when the $n^{\text{th}}$ code is used, the $i^{\text{th}}$ receiver can decode $\xb_i$ as $\Df_{(i-n)_N}(\cb_{R_{(i-n)_N}},\xb_{K_i})$. Thus, the codeword symbols queried by the $i^{\text{th}}$ receiver in the $n^{\text{th}}$ code is $|R_{(i-n)_N}| = m\,r_{(i-n)_N}$, where $m$ is the message length.

Now consider a time sharing scheme that utilizes each of the $N$ encoders $\Ef^{(1)},\dots,\Ef^{(N)}$ exactly once. The overall message length for this scheme is $mN$, the broadcast rate is $\beta$, and the number of codeword symbols queried by the $i^{\text{th}}$ receiver is 
\begin{equation*}
|R_i'| = \sum_{n \in [N] }|R_{(i-n)_N}| = m \sum_{ n \in [N]} r_n.
\end{equation*}
Similarly, for any $i \in [N]$ and $j = i \mod N + 1$, we have
\begin{align*}
|R_i' \cap R_j'| &= \sum_{n \in [N]} |R_{(i-n)_N} \cap R_{(i-n+1)_N}| \\
                 &= \sum_{n \in [N]} |R_{n} \cap R_{(n+1)_N}|,
\end{align*}  
which is independent of $i$.


\end{document}